\documentclass[a4paper,english,titlepage,12pt]{article}

\usepackage{babel}\usepackage{amsmath,amssymb,amsthm}
\usepackage[numbers,compress]{natbib}
\usepackage{graphicx}
\usepackage{longtable}
\usepackage{float}
\usepackage{titlesec}
\usepackage{caption}

\newtheorem{lemma}{Lemma}

\newtheorem{theorem}{Theorem}

\newcommand{\iclos}{\mathrm{IClos}}
\newcommand{\nloops}{\mathrm{Loop}}
\newcommand{\lepton}{\mathrm{Lept}}
\newcommand{\photon}{\mathrm{Ph}}
\newcommand{\prj}{\mathcal{A}}
\newcommand{\lpath}{\mathrm{LPath}}

\newcounter{appendcnt}

\makeatletter
\@addtoreset{equation}{section}
\@addtoreset{theorem}{section}
\@addtoreset{lemma}{section}
\makeatother

\begin{document}
\noindent\LARGE Infrared and Ultraviolet Power Counting on the Mass Shell in Quantum Electrodynamics
\\ \ \\ \ 
\noindent\large{\textbf{S. A. Volkov}}
\\ \ \\ \ 
\noindent\small{SINP MSU, Moscow, Russia; DLNP JINR, Dubna, Russia. E-mail: volkoff\underline{ }sergey@mail.ru, sergey.volkov.1811@gmail.com}
\\ \ \\ \ 
\normalsize
\noindent\textbf{Abstract.} A power counting rule is provided that allows us to obtain upper bounds for the absolute values of Feynman parametric integrands. The rule reflects both the ultraviolet and infrared behavior taking into account that the external momenta are on the mass shell. It gives us the ability to rigorously prove the absolute convergence of the corresponding integrals. The consideration is limited to the case of the quantum electrodynamics Feynman graphs contributing to the lepton magnetic moments and not containing either lepton loops or ultraviolet divergent subgraphs. However, a rigorous mathematical proof is given for all Feynman graphs satisfying these restrictions. The power counting rule is formulated in terms of Hepp's sectors, ultraviolet degrees of divergence and so-called I-closures. The obtained upper bound can not be substantially improved: the illustrative example is provided. The paper provides the first mathematically rigorous treatment of the ultraviolet behavior together with the on-shell infrared behavior with some kind of generality. Practical applications of this rule are explained.

\section{Introduction}\label{sec_intro}

A well known fact is that Feynman integrals written directly can contain ultraviolet (UV) and infrared (IR) divergences. In the case of the quantum electrodynamics (QED) the UV divergences can be recognized by  Dyson's power counting rule ~\cite{dyson}: a Feynman graph is free from UV divergences if this graph and all its subgraphs have the ultraviolet degree of divergence\footnote{see the definition in Section \ref{subsec_def}} less than zero. If a graph has UV divergences, counterterms are required to make the integral finite. However, Dyson's treatment has not been brought to the level of rigorous mathematical proofs. For the case of Euclidean propagators that are separated from zero denominators, rigorous proof of the power counting rule was obtained by S. Weinberg in 1960 ~\cite{weinberg_counting}. The UV divergences of Feynman integrals with genuine (Minkowski) propagators can be treated and subtracted with Bogoliubov's $\mathcal{R}$-operation that is, in a way, a development of A. Salam's ideas ~\cite{salam}. The corresponding theorem was proved by N. Bogoliubov and O. Parasiuk in 1956 ~\cite{bogolubovparasuk}. However, that proof contains some inaccuracies even for the case of Feynman graphs without UV-divergent subgraphs\footnote{For example, the inequality (4.20) in ~\cite{bogolubovparasuk} (or (20.1) in the Russian version) is not correct under the presupposed assumptions even for the scalar case.}. Despite these drawbacks the Bogoliubov-Parasiuk theorem established a quality standard and gave hope that quantum field theory has sense and can be rigorously examined. The flaws in the proof were corrected by K. Hepp in 1964 ~\cite{hepp}. In the QED case we have Feynman graphs with the propagators 
\begin{equation}\label{eq_propagators}
\frac{i(\hat{q}+m)}{q^2-m^2+i\varepsilon},\quad \frac{-g_{\mu\nu}}{q^2+i\varepsilon}
\end{equation}
for the lepton and photon lines correspondingly, where $\hat{q}=q^{\nu}\gamma_{\nu}$ for the Dirac gamma matrices $\gamma_{\nu}$. If we transfer to the Schwinger parameters by using the formula
$$
\frac{i}{x}=\int_0^{+\infty} e^{ixz} dz,
$$
we express the Feynman amplitude as an integral
\begin{equation}\label{eq_schwinger_parametric}
\int_{z_1,\ldots,z_L>0} I(z_1,\ldots,z_L,p_1,\ldots,p_n,\varepsilon) dz_1\ldots dz_L,
\end{equation}
where
\begin{equation}\label{eq_schwinger_product}
I(z,p,\varepsilon)=R(z,p)\cdot e^{iW(z,p,\varepsilon)},
\end{equation}
\begin{equation}\label{eq_znam}
W(z,p,\varepsilon)=\sum_{j,l} A_{jl}(z)p_jp_l - \sum_j z_j((m_j)^2-i\varepsilon),
\end{equation}
$p=(p_1,\ldots,p_n)$ is the vector of external momenta, $z=(z_1,\ldots,z_L)$ is the vector of Schwinger parameters, $R(z,p)$ is a polynomial in $p$ and is rational in $z$, $A_{jl}(z)$ is rational, homogeneous of degree $1$ and real, $m_j$ are the masses of the particles. The Bogoliubov-Parasiuk-Hepp theorem states that the $\mathcal{R}$-operation leads to an absolute convergent integral (\ref{eq_schwinger_parametric}), if $\varepsilon>0$ is fixed. Also, it was proved that for non-zero masses of the particles the Feynman amplitude tends to a distribution as $\varepsilon\rightarrow 0$. This result, however, is not concerned with real life despite the fact that there are some generalizations of this result in situations when zero-mass particles exist; the reason is that this does not take into account that the external momenta are on the mass shell. If a graph does not have UV-divergent subgraphs, then divergence subtractions are not needed; a direct proof in momentum space for the power counting rule was obtained by W. Zimmermann ~\cite{zimmermann_counting} with the help of a non-covariant regularization. Also, Zimmermann proved that the integrals tend to a relativistically covariant limit as the regularization parameter tends to zero\footnote{However, Feynman parameters were used for proving this.}. The asymptotic behavior of $I(z,p,\varepsilon)$ in (\ref{eq_schwinger_parametric}) may be estimated using a modification of E. Speer's lemma ~\cite{speer} for a fixed $\varepsilon>0$:
\begin{equation}\label{eq_speer}
|R(z,p)|\leq C(p)\cdot \frac{\left(z_{j_1}\right)^{d_1} \left(z_{j_2}/z_{j_1}\right)^{d_2} \left(z_{j_3}/z_{j_2}\right)^{d_3} \cdot \ldots\cdot \left(z_{j_L}/z_{j_{L-1}}\right)^{d_L}}{z_1\cdot z_2\cdot\ldots \cdot z_L},
\end{equation}
where $j_1,\ldots,j_L$ is the permutation of $1,\ldots,L$ that orders\footnote{These orders are called the Hepp sectors ~\cite{hepp}.} $z_j$:
\begin{equation}\label{eq_order}
z_{j_1}\geq z_{j_2}\geq\ldots\geq z_{j_L},
\end{equation}
$$
d_l=\lceil-\omega(\{j_l,j_{l+1},\ldots,j_L\})\rceil,
$$
$\omega(s)$ is the UV degree of divergence of the set of internal lines with the numbers in $s$ (see the definition in Section \ref{subsec_def}), $C$ is some coefficient, by $\lceil x\rceil$ we denote the minimum integer $y\geq x$. This estimation immediately leads to absolute convergence of (\ref{eq_schwinger_parametric}) for the given case\footnote{This can be proved by changing variables $t_1=z_{j_1}$, $t_2=z_{j_2}/z_{j_1}$, ..., taking into account that $d_j>0$ and $|e^{iW(z,p)}|\leq e^{-\varepsilon t_1}$.}. 

When the limit $\varepsilon\rightarrow 0$ is considered, we should transfer from the Schwinger parameters to the Feynman parameters by integrating analytically $I(z,p,\varepsilon)$ from (\ref{eq_schwinger_parametric}) with respect to $\lambda=z_1+\ldots+z_L$. If the sign of $W(z,p,0)$ in (\ref{eq_schwinger_product}) is constant, we obtain the Feynman amplitude as an integral
\begin{equation}\label{eq_feynman_parametric}
\int_{z_1,\ldots,z_L>0} I'(z_1,\ldots,z_L,p) \delta(z_1+\ldots+z_L-1) dz_1\ldots dz_L,
\end{equation}
where $I'(z,p)$ is the sum of terms of the form
\begin{equation}\label{eq_feynman_param_term}
\frac{R'(z,p)}{W(z,p)^M},
\end{equation}
where $R'$ is a polynomial in $p$ and is rational in $z$, (\ref{eq_znam}) is satisfied with $W(z,p)=W(z,p,0)$, $M>0$ is some integer, if the considered graph has a negative UV degree of divergence or the corresponding ``overall-divergent'' terms are cancelled in some way. The numerator $R'(z,p)$ can be estimated analogously to (\ref{eq_speer}). However, the existence of the denominator $W(z,p)^M$ significantly complicates the problem. In the general case, this leads to an emergence of IR divergences. These divergences cancel in scattering cross-sections if we take into account the finite sensitivity of photon detectors and if we treat correctly the IR divergences connected with the lepton wave function renormalization. These cancellations were considered by different authors; see, for example, ~\cite{bloch_nordsieck,jauch_rohrlich,yennie} and ~\cite{weinberg_textbook}\footnote{The IR divergences cancellation is explained  in Part 13 ``Infrared Effects'' of ~\cite{weinberg_textbook}.}. However, the level of rigour in these examinations is too far from the quality standard established by the Bogoliubov-Parasiuk-Hepp theorem: all of these treatments use unjustified approximations\footnote{For example, the IR divergences of power type that can emerge when we have a lepton self-energy subgraph require a more meticulous treatment; see Discussion in ~\cite{volkov_2015} and the references in it. However, these divergences are not the only thing that require an accurate consideration: all estimations, in principle, require the justification.}; the matter of IR divergences connected with the on-shell renormalization constants is ignored or is mistakenly taken under examination (despite the fact that these IR divergences are the only reason why the lepton magnetic moments are IR-finite in QED\footnote{see Part C of ~\cite{grozin}}). The cancellation of IR divergences is often connected with the Kinoshita-Lee-Nauenberg (KLN) theorem. However, it is only a tradition to use this denomination. The letter ``K'' from KLN corresponds to ~\cite{kinoshita_singularities}. That paper contains a detailed description of the nature of IR singularities and their cancellation in terms of Feynman parameters. However, that consideration does not try to meet the standards of mathematical rigour and does not contain any theorems. Moreover, the 3-loop calculation of the electron anomalous magnetic moment performed by T. Kinoshita and P. Cvitanovi\`{c} in 1974 required a treatment that is beyond the scope of ~\cite{kinoshita_singularities}; see ~\cite{kinoshita_infrared}. The letters ``LN'' from KLN correspond to ~\cite{lee_nauenberg} with a quantum-mechanical treatment of the problem. However, that treatment requires a correct formulation of quantum field theory in terms of the evolution of the state vector over time; this is a much stronger presupposition than any finiteness theorem about scattering cross-sections in the perturbation theory. Thus, that result can be considered only as a substantive treatment at a physical level of rigour. Note that there are some extensions of the $\mathcal{R}$-operation and Bogoliubov-Parasiuk-Hepp theorem to the case of IR divergences; see, for example, ~\cite{karchev,lowenstein,rstar,bergerelam}. However, these extensions are not related to the on-shell IR divergences that occur in real life.

The case of the lepton magnetic moment is a relatively simple one: all IR divergences are removed by the on-shell renormalization with the UV ones; consideration of real soft photon emission is not required\footnote{However, for obtaining an absolutely convergent integral for each Feynman graph it is required to subtract IR divergences together with the UV ones. Different authors developed different methods for performing this subtraction that work in some cases. See ~\cite{kinoshita_infrared,carrollyao,levinewright,kinoshita_atoms,volkov_2015}. Let us emphasize that all those methods eliminate divergences directly in parametric spaces without any use of dimensional regularization and so on (in the spirit of ~\cite{bogolubovparasuk,hepp}).}. Moreover, the sign of $W(z,p,0)$ in (\ref{eq_schwinger_product}) is constant inside the integration area. However, $W(z,p)$ in (\ref{eq_feynman_param_term}) vanishes on some points of the boundary, and a rigorous finiteness proof still has not been obtained even for this case. Let us draw attention to some difficulties that appear when we investigate IR divergences. All UV divergences in the Schwinger parametric integral (\ref{eq_schwinger_parametric}) can be recognized by taking $z_{j_1}=z_{j_2}=\ldots=z_{j_l}=\delta$ and calculating the leading power of $\delta$ in the limit $\delta\rightarrow 0$. However, this statement is wrong for IR divergences in Feynman parametric integrals even for constant-sign $W(z,p)$. For example, sometimes a substitution of the form $z_{j_1}=\ldots=z_{j_l}=\delta,z_{j_{l+1}}=\ldots=z_{j_q}=\delta^2$ is required; see ~\cite{kinoshita_infrared}. Also, IR and UV divergences can occur in a mixed form. 

This can be illustrated by the simple integral
$$
f(\lambda,\Lambda)=\int_{\lambda}^{\Lambda} \frac{dxdy}{x^4y^2+1},
$$
$0<\lambda<\Lambda$. It is easy to see that $\lim_{\lambda\rightarrow 0} f(\lambda,\Lambda)$ is finite for any fixed $\Lambda$ as well as $\lim_{\Lambda\rightarrow +\infty} f(\lambda,\Lambda)$ is finite for any fixed $\lambda>0$. However, the simultaneous limit $\lim_{\lambda\rightarrow 0,\Lambda\rightarrow +\infty} f(\lambda,\Lambda)$ is infinite\footnote{The integral over the domain $\{x\geq\lambda;y>0\}$ equals $C\int_{\lambda}^{+\infty} dx/x^2 = C/\lambda$, where $C=\int_0^{+\infty} dt/(t^2+1)$, but over $\{x>0;y\leq\Lambda\}$ it equals $C\int_{0}^{\Lambda}dx/\sqrt{x}=2C\sqrt{\Lambda}$, where $C=\int_0^{+\infty} dt/(t^4+1)$.}. For a physical case this mixing can be demonstrated by explicit formulas in ~\cite{adkins}. In addition, it is not obvious from the beginning that the denominator in (\ref{eq_feynman_param_term}) will not spoil the convergence corresponding to the numerator, if the numerator is ``near to the UV-divergent'' (for example, if $d_j>0$ in (\ref{eq_speer}) are near to zero). Therefore, we must consider the UV and IR behavior together. The author believes that these difficulties are one of the main reasons why the problem of rigorous treatment of IR divergences lost the attention of scientists starting in the 1980's and why most computations in quantum field theory use methods that allow us to calculate without understanding the nature of IR divergences\footnote{based on analytical continuations on the complex plane and, in particular, dimensional regularization}.

In this paper we obtain an upper bound for $|I'(z,p)|$ in (\ref{eq_feynman_parametric}) in the form (\ref{eq_speer}) with another degrees. This estimation provides a rigorous proof of the absolute convergence of the Feynman parametric integral. The consideration is restricted to the case of QED Feynman graphs contributing to the lepton magnetic moment and not containing lepton loops and UV-divergent subgraphs. However, even for this case a rigorous convergence proof had not been obtained before. The obtained estimation is based on the notion of \emph{I-closures} that was introduced in ~\cite{volkov_prd} (see the definition in Section \ref{subsec_def}). More precisely, we prove that the estimation of the form (\ref{eq_speer}) is satisfied with
$$
d_l=\max\left( \lceil -\omega(\iclos(\{j_l,\ldots,j_L\}))\rceil - \frac{1}{2}, \frac{1}{2}\right).
$$
The I-closure $\iclos$ gives a set with the greater ultraviolet degree of divergence. This allows us to take into account the presence of a denominator in (\ref{eq_feynman_param_term}). As a consequence, it is easy to see for the Feynman parametric integrand $I'(z)$ that
\begin{equation}\label{eq_general_est}
|I'(z)|\leq C\cdot \frac{(\min(z_1,\ldots,z_L))^{1/2}}{z_1\cdot\ldots \cdot z_L}.
\end{equation}
This leads to the absolute convergence of the Feynman parametric integral. The power $1/2$ in (\ref{eq_general_est}) is unimprovable (in contrast to the known case with fixed $\varepsilon>0$ where all degrees are integer, and $1/2$ can be replaced with $1$); see an example in Appendix. The considered set of Feynman graphs, on the one hand, possesses some kind of generality and allows us to demonstrate ideas and a technique of how the on-shell IR behavior can be treated rigorously together with the UV behavior; on the other hand, the proofs are relatively simple for this case\footnote{for example, in comparison with known proofs of the Bogoliubov-Parasiuk-Hepp theorem}. The paper provides the first mathematically rigorous proof that all divergences in the graphs without lepton loops contributing to the lepton magnetic moment are connected with the UV divergent subgraphs.

In addition to rigorous proofs, the obtained estimations for the integrands have a practical application: it can be used for constructing the probability density functions (PDF) for Monte Carlo integration\footnote{See a review about Monte Carlo integration in ~\cite{mc_james}.}. For integration of a function $f(z)$ using Monte Carlo the PDF $g(z)$ should be as near as possible to $C|f(z)|$ (for reaching a fast convergence and a high precision of integration) and suitable for random sampling. When $f(z)$ is not square-integrable and the number of variables is large, the PDF should be chosen very accurately: on the one hand, an underestimation of the asymptotic growth rate near the boundary even for one variable and for one sector leads to an infinite value $\int \frac{f(z)^2}{g(z)}dz$ and to poor and unstable convergence (as a consequence); on the other hand, a total overestimation of this growth rate leads to poor convergence too. Estimations like (\ref{eq_speer}) with different asymptotic growth rates for different Hepp sectors are useful in this case. PDFs based on similar expressions with the use of I-closures exploited for calculating the total contribution of the 5-loop QED Feynman graphs without lepton loops to the electron anomalous magnetic moment\footnote{This contribution is sensible in experiments and still does not have a reliable computed value ~\cite{volkov_5loops_prd}. For the 5-loop case, known computation methods that allow us to ``calculate without thinking'' require an enormous amount of computer time.} ~\cite{volkov_5loops_prd}. Those PDFs allowed us to reach the needed precision of the Monte Carlo integration on a relatively small supercomputer. However, more complicated expressions were used for that development: divergence subtraction was required, and some adjustment was needed for reaching the maximum Monte Carlo convergence speed.

The paper is organized as follows. Known rules for constructing Schwinger parametric and Feynman parametric integrands are given in Section \ref{subsec_constr} after definitions in Section \ref{subsec_def}. A modification of E. Speer's lemma ~\cite{speer} that is convenient for our examinations is proved in Section \ref{subsec_numerator}. It was noted that the numerator in (\ref{eq_feynman_param_term}) is easily estimated using the known combinatorial formulas with sums over graphs. However, the influence of the denominator is processed better with the help of the electric circuit analogy; Section \ref{subsec_denominator} explains this influence. The impact of the magnetic moment projector\footnote{It is known from many quantum field theory textbooks that the magnetic moment projector eliminates all IR and UV divergences in the 1-loop case. However, it plays an important role in removing divergences for the higher orders too.} can be taken into account in both ways; we use the combinatorial one. This impact is described in Section \ref{subsec_projector} with the completion of the main estimation proof. The appendix contains an example demonstrating that the obtained estimation is exact.

\section{Preliminaries}

\subsection{Definitions}\label{subsec_def}

We use the metric $g_{\mu\nu}$, where $g_{00}=1$, $g_{11}=g_{22}=g_{33}=-1$, the Dirac gamma matrices $\gamma_{\mu}$ satisfy $\gamma^{\mu}\gamma^{\nu}+\gamma^{\nu}\gamma^{\mu}=2g^{\mu\nu}$.

Let $G$ be a QED Feynman graph with the propagators (\ref{eq_propagators}) for lepton and photon lines. We suppose that $G$ is strongly connected, has two external lepton lines and one external photon line, $L$ internal lines that are enumerated $1,2,\ldots,L$. By $V(G)$ and $E(G)$ we denote the sets of all vertexes and internal lines of $G$ correspondingly. The Feynman amplitude $\Gamma_{\mu}(p,q)$ corresponds to $G$; here $p-q/2$ is the incoming external lepton momentum, $p+q/2$ is the outgoing one, $q$ is the external photon momentum.

Let us define the \emph{anomalous magnetic moment projector} $\prj$. The Feynman amplitude $\Gamma(p,q)$ can be expressed in terms of three form-factors:
$$
\overline{\psi}_2 \Gamma_{\mu}(p,q) \psi_1 = \overline{\psi}_2 \left( f(q^2)\gamma_{\mu} - \frac{1}{2m}g(q^2) \sigma_{\mu\nu} q^{\nu} + h(q^2) q_{\mu} \right) \psi_1,
$$
where
$$
\left( \hat{p}-\frac{\hat{q}}{2} - m\right) \psi_1 = \overline{\psi}_2 \left( \hat{p}+\frac{\hat{q}}{2} - m\right)=0,
$$
$$
\sigma_{\mu\nu}=\frac{1}{2} (\gamma_{\mu}\gamma_{\nu}-\gamma_{\nu}\gamma_{\mu}).
$$
Put 
$$
\prj\Gamma=\lim_{q^2\rightarrow 0} g(q^2)
$$
(see ~\cite{ll4}\footnote{part XII ``Radiative corrections'', \S 116 ``Electromagnetic form factors of the electron''; note that the term $h(q^2) q_{\mu}$ vanishes in the final result, but it remains non-zero in the contributions of separate graphs and $z$.}).

A set $s\subseteq E(G)$:
\begin{itemize}
\item
 is called an \emph{1-tree}, if there is a path in $s$ between any $v_1,v_2\in V(G)$ and $s$ does not have cycles;
\item is called a \emph{tree with cycle} if it has a cycle and it becomes a 1-tree after deleting any line of this cycle.
\end{itemize}

By $\lepton(s)$ and $\photon(s)$ we denote the sets of all lepton and photon lines of $s$ correspondingly; by $\nloops(s)$ we denote the minimum of $|s\backslash T|$, where $T$ runs over all 1-trees of $G$.

The \emph{ultraviolet degree of divergence} is defined as
$$
\omega(s)=2\cdot\nloops(s)-|s|+\frac{1}{2}|\lepton(s)|.
$$
For connected sets $s$ it equals $2-\frac{3}{4}N_{\lepton}-\frac{1}{2}N_{\photon}$, where $N_{\lepton}$, $N_{\photon}$ are the numbers of external lepton and photon lines in regard to $s$.

We suppose that $G$ does not have lepton cycles, and $\omega(s)<0$ for all $s\subseteq E(G)$ except the empty set and $E(G)$.

We say that a vertex $v$ is \emph{incident} to a line $l$, if $v$ is one of the ends of $l$.

If $i\in\photon(E(G))$, then by $\lpath(i)$ we denote the set of all lines that are on the lepton path connecting the vertexes incident to $i$.

By definition, put
$$
\iclos(s)=s\cup\{i\in\photon(E(G)):\ \lpath(i)\subseteq s\}.
$$
For example, for the graph from Fig. \ref{fig_example_iclos},
$$
\iclos(\{3,5,6\})=\{3,5,6\},
$$
$$
\iclos(\{3,4,5,6\})=\{3,4,5,6,9\},
$$
$$
\iclos(\{2,3,4,5,6,7\})=\{2,3,4,5,6,7,8,9\},
$$
$$
\iclos(\{1,2,3,4,5,6\})=\{1,2,3,4,5,6,7,8,9\}.
$$

\begin{figure}[h]
\begin{center}
\includegraphics{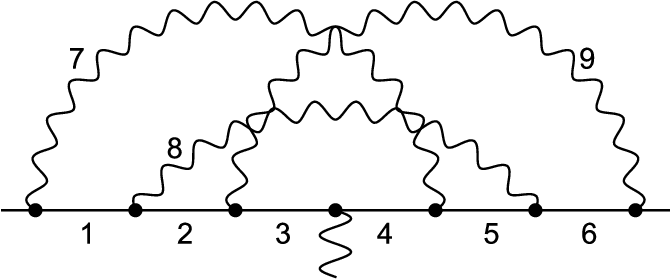}
\end{center}
\caption{Example of a Feynman graph contributing to the lepton anomalous magnetic moment.}
\label{fig_example_iclos}
\end{figure}

\subsection{Construction of the integrands}\label{subsec_constr}

\subsubsection{General rules in Schwinger parametric space}\label{subsec_constr_schwinger}

Let us first describe the known rules of constructing (\ref{eq_schwinger_parametric}) corresponding to a general QED Feynman graph $G$ with the external momenta $p_1,\ldots,p_n$ without lepton loops. The rules are based on the known formulas; see ~\cite{zavialov,smirnov}\footnote{More precisely, the formulas are described in Chapter II ``Parametric representations for Feynman diagrams. R-operation'', Section 1 ``Regularized Feynman diagrams'' of ~\cite{zavialov} and in Part I "Regularized Feynman amplitudes", Section 2 ``Parametric representations of Feynman amplitudes'' of ~\cite{smirnov}; we use those formulas combined with the Feynman rules of QED. See also the formulas in ~\cite{kinoshita_rules,kinoshita_automated} formulated in a slightly different way.}. The coefficient before the expression is omitted.
\begin{itemize}
\item $R(z,p)$ from (\ref{eq_schwinger_product}) is represented as
\begin{equation}\label{eq_r_details}
R(z,p)=\Pi_1(p)R_1(z)+\ldots+\Pi_N(p)R_N(z),
\end{equation}
where $R_j(z)$ is rational and homogeneous of degree $u_j$, each $\Pi_j(p)$ is the product of multipliers of the form $\hat{p}_l$, $p_{l\mu}$, $\gamma_{\mu}$, $g_{\mu\nu}$, $(p_{l'}p_{l''})$, where $l,l',l''$ are indexes of the external lines, $\mu,\nu$ are tensor indexes of some external photon lines. Tensor indexes are omitted for simplicity.
\item Each term $\Pi_jR_j$ corresponds to a set $P_j$ of disjoint unordered pairs of internal lepton lines\footnote{The summation over sets of pairs comes from the formula of multiple differentiation of $\exp(f(\xi))$ with respect to $\xi$, where $f(\xi)$ is a polynomial of degree $2$. The variables $\xi$ come from the additional multipliers $\exp(q_l\xi_l)$ to the Schwinger parametric propagators $\exp(iz_l(q_l^2-m^2+i\varepsilon))$, where $q_l$ is the momentum of the internal line $l$. These multipliers are used for obtaining $\hat{q}$ in the numerator of the lepton propagator (\ref{eq_propagators}) by differentiating with respect to $\xi$ at $\xi=0$ and lead to a modification of $W(z,p,\varepsilon)$ in (\ref{eq_schwinger_product}). See ~\cite{zavialov,smirnov} and also ~\cite{kinoshita_rules}.}. One set of pairs can generate several terms $\Pi_jR_j$; however, all $R_j(z)$ corresponding to one set have the same degree.
\item The part of (\ref{eq_r_details}) that corresponds to a set $P$ of pairs of lepton lines is obtained by multiplying (and performing the needed tensor convolutions using $g^{\mu\nu}$) of the multipliers that are described below. The order of the multiplication corresponds to Feynman rules: if $a_1,\ldots,a_n$ are the multipliers are ordered along the path ($a_1$ and $a_n$ correspond to the beginning and end accordingly), then the corresponding product is $a_n a_{n-1}\ldots a_1$. The coefficient is omitted. An internal lepton line is called \emph{paired} if it belongs to some pair in $P$. The other internal lepton lines are called \emph{unpaired}.
\item The global multiplier 
$$
\frac{1}{D(z)^2}
$$
starts the product, where
$$
D(z)=\sum_T \prod_{l\in E(G)\backslash T} z_l,
$$
the summation goes over all 1-trees in $G$.
\item The multiplier $\gamma_{\nu}$ corresponds to each of the graph vertex, where $\nu$ is the tensor index that corresponds to the photon line that is incident to this vertex.
\item The multiplier $\gamma_{\nu}$ also corresponds to each paired internal lepton line, where $\nu$ is the tensor index corresponding to the pair in $P$ that contains this line.
\item Also, each pair $\{l_1,l_2\}\in P$ gives the global multiplier 
$$
\frac{B_{l_1l_2}(z)}{D(z)},
$$
where
\begin{equation}\label{eq_b}
B_{l_1l_2}(z)=\sum_{T'} c(T')\prod_{l\in E(G)\backslash T'} z_l,
\end{equation}
the summation goes over all trees with cycle containing $l_1$ and $l_2$ on the cycle, $c(T')=1$ if the direction of $l_1$ and $l_2$ on the cycle is the same, $c(T')=-1$ if the direction is opposite.
\item The multiplier 
$$
m+\frac{\hat{Q}_l(z)}{D(z)}
$$
corresponds to each unpaired internal lepton line $l$. Here
$$
Q_l(z)=\sum_T p[T] \prod_{l'\in E(G)\backslash T} z_{l'},
$$
where the summation goes over all 1-trees $T$ containing $l$, by $p[T]$ we denote the momentum that passes through $l$ in $T$; it is defined as the sum of the external momenta of $G$ outgoing from the connectivity component of $T\backslash l$ to which the line $l$ is directed.
\item $W(z,p,\varepsilon)$ from (\ref{eq_schwinger_product}) is obtained using the electric circuit analogy ~\cite{bjorkendrell2}\footnote{Chapter 18 ``Dispersion Relations'', Section 18.4 ``Generalization to Arbitrary Graphs and the Electrical Circuit Analogy'' of ~\cite{bjorkendrell2}}. Namely, $\sum_{j,l}A_{jl}p_jp_l$ from (\ref{eq_znam}) equals the electrical power that is absorbed within the circuit with topology $G$, line resistances $z_l$, external currents\footnote{The coefficients $A_{jl}(z)$ should be obtained supposing that the external currents are real numbers. After this, the 4-vectors $p_1,\ldots,p_n$ are substituted to the formula.} $p_1,\ldots,p_n$.
\end{itemize}

\subsubsection{Feynman parameters and the magnetic moment projector}\label{subsec_constr_feynman_amm}

Let $G$ be a Feynman graph satisfying the restrictions from Section \ref{subsec_def}; its external momenta are determined by two parameters $p$ and $q$.

Let us consider the Schwinger parametric term
$$
\Pi_j(p,q)R_j(z)e^{iW(z,p,q,\varepsilon)}.
$$
After applying the projector $\prj$ for $\varepsilon=0$ we obtain
$$
\prj\left[\Pi_j(p,q)R_j(z)e^{iW(z,p,q,0)}\right]=[\prj\Pi_j]R_j(z)e^{iW(z)},
$$
where 
\begin{equation}\label{eq_w_short_def}
W(z)=\left.W(z,p,q,0)\right|_{p^2=m^2,q^2=0}.
\end{equation}

The equality follows from the fact that $\prj[\Gamma(p,q)]$ is determined by $\Gamma(p,q)$ at $\{p^2=m^2,q=0\}$ and its derivatives at these points along the directions\footnote{see an explicit formula in ~\cite{kinoshita_automated}} $(p',q')$ satisfying the condition $pp'=pq'=0$; these derivatives of $W$ equal $0$, because they equal $0$ for all scalar products of the external momenta. 

Let us remark that
\begin{equation}\label{eq_w_from_z}
W(z)=m^2(Z(z)-Z_0(z)),
\end{equation}
where $Z(z)$ is the resistance between the vertexes that are incident to the external lepton lines of the electric circuit with topology $G$, line resistances $z_l$; $Z_0(z)=\sum_{l\in\lepton(E(G))} z_l$ is the resistance of the corresponding circuit with removed photon lines.

If $R_j(z)$ is homogeneous of degree $u_j$, then after transferring to the Feynman parameters we obtain the corresponding term
$$
\frac{[\prj\Pi_j]R_j(z)}{W(z)^{L+u_j}}
$$
(the coefficient is omitted). Let us note that the analytical integration with respect to $\lambda=z_1+\ldots+z_L$ is not applicable to the ``overall'' UV-divergent terms, i.e. the terms corresponding to the case when all lines from $\lepton(E(G))$ are paired; however, these terms are eliminated by the magnetic moment projector, because
$$
\prj \gamma_{\mu}=0.
$$

\section{Proof of the main inequality}

\subsection{Estimation of the numerator}\label{subsec_numerator}

Let us consider the sum of the Feynman parametric terms corresponding to the set $P$ of pairs of internal lepton lines. This sum
$$
\sum_{j:P_j=P} \frac{[\prj\Pi_j]R_j(z)}{W(z)^{L+u_j}}
$$
equals
\begin{equation}\label{eq_global_multiplier_factor_out}
\frac{K(z)\cdot \sum_{j:P_j=P} [\prj \Pi_j]Y_j(z)}{W(z)^{L+u[P]}},
\end{equation}
where $K(z)$ contains all global multipliers $(1/D(z)^2)$ and $(B_{l_1l_2}(z)/D(z))$ of the construction, $u[P]$ equals $p_j$ for all $j$ such that $P_j=P$ (they are the same), $Y_j(z)$ absorbs the remaining multipliers in $R_j(z)$.

By definition, put
$$
P[s]=\{a\in P:\ a\subseteq s\},
$$
for sets $s\subseteq E(G)$.

Let $j_1,\ldots,j_L$ be the permutation of $1,\ldots,L$ satisfying (\ref{eq_order}). By $s^{[l]}$ we denote the set $\{j_l,j_{l+1},\ldots,j_L\}$. Also, we will use the notation
$$
t_1=z_{j_1},\ t_2=\frac{z_{j_2}}{z_{j_1}},\ 
\ldots,\ t_L=\frac{z_{j_L}}{z_{j_{L-1}}}.
$$

We will also use the fact that $z_1,\ldots,z_L$ are Feynman parameters, i.e., $z_1+\ldots+z_L=1$. Therefore, $t_1,\ldots,t_L\leq 1$.

\begin{lemma}\label{lemma_prod}
For any $T\subseteq E(G)$ we have
$$
\prod_{l\in E(G)\backslash T} z_l = \prod_{l=1}^L t_l^{|s^{[l]}\backslash T|}.
$$
\end{lemma}
\begin{proof}
$$
\prod_{l\in E(G)\backslash T} z_l = \prod_{j_l\in E(G)\backslash T} z_{j_l} = \prod_{j_l\in E(G)\backslash T} (t_1t_2\ldots t_l) = \prod_{l=1}^L t_l^{|s^{[l]}\backslash T|}.
$$
\end{proof}

\begin{lemma}\label{lemma_d}
The following inequality is satisfied\footnote{Also, it can be proved that the inequality is satisfied in both directions: the sign $\geq$ can be changed to $\leq$ (however, this fact is not needed for our examination).}:
$$
D(z)\geq C \prod_{l=1}^L t_l^{\nloops(s^{[l]})},
$$
where $C>0$ is some constant (depending only on the structure of the graph).
\end{lemma}
\begin{proof}
Taking into account Lemma \ref{lemma_prod}, we should only prove that there exists a 1-tree $T$ such that $|s^{[l]}\backslash T|=\nloops(s^{[l]})$ for all $l$. This 1-tree can be obtained by adding lines (starting from the empty set) to complement the set to a maximal acyclic set in $s^{[l]}$ sequentially for $l=L,L-1,\ldots,1$.
\end{proof}

\begin{lemma}\label{lemma_speer}
$K(z)$ from (\ref{eq_global_multiplier_factor_out}) satisfies
$$
|K(z)|\leq C \cdot\frac{\prod_{l=1}^L t_l^{|s^{[l]}|-2\cdot \nloops(s^{[l]})-|P[s^{[l]}]|}}{z_1\ldots z_L}.
$$
for some constant $C$ (depending only on the structure of the graph $G$).
\end{lemma}
\begin{proof}
Let us first remark that the right part equals
$$
C\cdot\prod_{l=1}^L t_l^{-2\cdot \nloops(s^{[l]})-|P[s^{[l]}]|}.
$$
Lemma \ref{lemma_d} states that 
$$
\frac{1}{D(z)^2} \leq C\cdot \prod_{l=1}^L t_l^{-2\cdot \nloops(s^{[l]})}
$$
for some $C$. 
Let us consider the contribution of the multipliers
\begin{equation}\label{eq_speer_lemma_b}
\frac{B_{j_1j_2}(z)}{D(z)}.
\end{equation}
Let us fix $j_1,j_2$ and $T'$ from (\ref{eq_b}). From Lemmas \ref{lemma_prod} and \ref{lemma_d} it follows that this tuple gives a contribution not exceeding (in absolute value)
$$
C\cdot\prod_{l=1}^L t_l^{a_l}
$$
for some $C$ to (\ref{eq_speer_lemma_b}), where 
$$
a_l=|s^{[l]}\backslash T'|-\nloops(s^{[l]}).
$$
If $j_1,j_2\in s^{[l]}$, then $T'$ becomes acyclic in $s^{[l]}$ after removing $j_1$; thus $a_l\geq -1$. If $j_1$ or $j_2$ is not in $s^{[l]}$, then the cycle of $T'$ is not contained in $s^{[l]}$, i.e., $T'\cap s^{[l]}$ is acyclic, $a_l\geq 0$.
This completes the proof.
\end{proof}

\subsection{Influence of the denominator}\label{subsec_denominator}

\begin{lemma}\label{lemma_w}
The following inequality is satisfied\footnote{This inequality works in both directions too.} for (\ref{eq_w_short_def}):
\begin{equation}\label{eq_lemma_w}
|W(z)|\geq C\cdot \max_{i\in \photon(E(G))} \frac{(z'_i)^2}{\max(z'_i,z_i)},
\end{equation}
where 
$$
z'_i=\max_{l\in\lpath(i)} z_l,
$$
$C>0$ is some constant depending only on the structure of the graph (and $m$).
\end{lemma}
\begin{proof}
Take $i$ on which the maximum (\ref{eq_lemma_w}) is reached. By $Z_i(z)$ we denote the resistance between the points incident to the external lepton lines of the electric circuit with the graph that is obtained from $G$ by removing all photon lines except $i$ (and the resistances $z_l$).

It is obvious that
$$
Z(z)\leq Z_i(z)\leq Z_0(z)
$$
(we use $Z$ and $Z_0$ from (\ref{eq_w_from_z})).
Thus,
$$
|W(z)|\geq m^2(Z_0(z)-Z_i(z))=m^2\left(z''_i-\frac{z_i z''_i}{z_i+z''_i}\right) = m^2\frac{(z''_i)^2}{z''_i+z_i},
$$
where
$$
z''_i=\sum_{l\in\lpath(i)} z_l.
$$
The proof is completed taking into account that 
$$
\max_{i=1}^n a_i\leq \sum_{i=1}^n a_i\leq n\cdot \max_{i=1}^n a_i
$$ for any $a_1,\ldots,a_n\geq 0$.
\end{proof}

Let us describe an idea how to transform the estimation from Lemma \ref{lemma_w} to the form like $\prod t_l^{d_l}$. If $r_i\geq 0,i\in \photon(E(G))$ are some real numbers, $M=\sum_i r_i$, then we have
\begin{equation}\label{eq_denominator_prod}
|W(z)|^M\geq C\cdot \prod_{i\in \photon(E(G))} \left(\frac{(z'_i)^2}{\max(z'_i,z_i)}\right)^{r_i} = C\cdot\prod_{l=1}^L t_l^{a_l},
\end{equation}
where
$$
a_l=\sum_{i\in\photon(\iclos(s^{[l]}))} r_i + \sum_{i\in\iclos(s^{[l]})\backslash s^{[l]}} r_i.
$$

We continue to consider the terms corresponding to the fixed set $P$ of pairs and to use the definitions from Section \ref{subsec_numerator}.

\begin{lemma}\label{lemma_a_b}
For this selection of $r_i$ and $M$ the following inequality is satisfied:
\begin{equation}\label{eq_lemma_a_b}
\left|\frac{K(z)}{W(z)^M}\right|\leq C \frac{\prod_{l=1}^L t_l^{\lceil-\omega(\iclos(s^{[l]}))\rceil+A_l+B_l}}{z_1\ldots z_L},
\end{equation}
where
\begin{equation}\label{eq_al}
A_l=\lfloor \frac{|\lepton(s^{[l]})|}{2} \rfloor - |P(s^{[l]})| - \sum_{i\in\photon(\iclos(s^{[l]}))} r_i,
\end{equation}
\begin{equation}\label{eq_bl}
B_l=\sum_{i\in\iclos(s^{[l]})\backslash s^{[l]}} (1-r_i),
\end{equation}
by $\lfloor x\rfloor$ we denote the maximum integer $y\leq x$.
\end{lemma}
\begin{proof}
The power of $t_l$ is constituted of the following terms:
\begin{itemize}
\item 
$$
\lceil-\omega(s^{[l]})\rceil + \lfloor |\lepton(s^{[l]})|/2 \rfloor - |P[s^{[l]}]|
$$ 
(Lemma \ref{lemma_speer}; this lemma also gives the denominator of the right part of (\ref{eq_lemma_a_b});
\item  
$$
-\sum_{i\in\photon(\iclos(s^{[l]}))} r_i - \sum_{i\in\iclos(s^{[l]})\backslash s^{[l]}} r_i
$$
(from (\ref{eq_denominator_prod}).
\end{itemize}
To conclude the proof it remains to note that for any $s\subseteq E(G)$ we have
$$
\omega(\iclos(s))=\omega(s)+2(\nloops(\iclos(s))-\nloops(s))-|\iclos(s)\backslash s|$$ $$=\omega(s) + |\iclos(s)\backslash s|
$$
(because the I-closure does not change the number of connectivity components of the set).
\end{proof}

\begin{lemma}\label{lemma_photon_iclos_ineq}
For any $s\subseteq E(G)$ we have
$$
|\photon(\iclos(s))|\leq \lfloor |\lepton(s)|/2 \rfloor.
$$
\end{lemma}
\begin{proof}
It is enough only to prove that 
$$
|\photon(\iclos(s))| \leq |\lepton(s)|/2.
$$
For proving this we consider $\lepton(s)$ as the union of paths $h_1,\ldots,h_l$ (that are not connected). There are two cases:
\begin{itemize}
\item For each path $h_j$ there exists a photon line of $G$ (internal or external) with exactly one vertex on $h_j$. Taking to account that $h_j$ has $|h_j|+1$ vertexes, and one vertex is already occupied by that photon, there are no more than $|h_j|/2$ photons with ends on $h_j$ contributing to $|\photon(\iclos(s))|$.
\item 
There exists a path $h_j$ without external photon lines in regard to it. Then $\iclos(h_j)$ forms a lepton self-energy subgraph\footnote{If $G$ contains lepton self-energy subgraphs, then all considerations with UV degrees of divergence crash down: IR divergences of power type emerge in this case; this contradicts to all estimations based on UV degrees of divergence; these divergences transform to the logarithmic ones after subtraction of the self-mass part; see Discussion in ~\cite{volkov_2015} and the references in it.}, i.e, $\omega(\iclos(h_j))=1/2>0$; this contradicts to the restrictions on $G$. 
\end{itemize}
\end{proof}

\begin{lemma}\label{lemma_r}
There exist numbers $r_i\in\{0,1\}$ for $i\in \photon(E(G))$ such that $\sum_i r_i=L+u[P]$ and for any $l=1,\ldots,L$ we have
\begin{equation}\label{eq_lemma_r}
\sum_{i\in \photon(\iclos(s^{[l]}))} r_i \leq \lfloor \lepton(s^{[l]})/2 \rfloor - |P[s^{[l]}]|.
\end{equation}
\end{lemma}
\begin{proof}
First, we note that 
$$
L+u[P]=|\photon(E(G))|-|P|
$$ (by construction, see Section \ref{subsec_constr_schwinger}).
Taking into account Lemma \ref{lemma_photon_iclos_ineq}, it is sufficient to construct $r$ satisfying
\begin{equation}\label{eq_lemma_r_modified}
\sum_{i\in \photon(\iclos(s^{[l]}))} (1-r_i) \geq \min(|P[s^{[l]}]|, |\photon(\iclos(s^{[l]}))|)
\end{equation}
instead of (\ref{eq_lemma_r}). To construct $r$ we start from $r_i=1$ for all $i$ and perform the following operation sequentially for $l=L,L-1,\ldots,1$:
\begin{itemize}
\item if (\ref{eq_lemma_r_modified}) is not satisfied for $l$, then change some $r_i$ for $i\in\photon(\iclos(s^{[l]}))$ to zero to satisfy
$$
\sum_{i\in \photon(\iclos(s^{[l]}))} (1-r_i) = \min(|P[s^{[l]}]|, |\photon(\iclos(s^{[l]}))|).
$$
\end{itemize}
It is clear that performing this operation for a given $l$ will not spoil the verity of (\ref{eq_lemma_r_modified}) for the greater $l$ (because $r_i$ are only decreased). Also, it is easy to see that the following condition is satisfied after performing the operation for $l$: $r_i=1$ for all $i\notin \photon(\iclos(s^{[l]}))$ and $\sum_i (1-r_i)\leq |P|$ (this can be proved by induction with respect to $l$). This guarantees that $\sum_i r_i=|\photon(E(G))|-|P|$ after performing all operations (because $|P|\leq |\photon(E(G))|$).
\end{proof}
Lemmas \ref{lemma_a_b} and \ref{lemma_r} lead to
$$
\left|\frac{K(z)}{W(z)^{L+u[P]}}\right|\leq C \frac{\prod_{l=1}^L t_l^{\lceil-\omega(\iclos(s^{[l]}))\rceil}}{z_1\ldots z_L}.
$$
However, this is insufficient for proving the convergence of the integral, because the powers must be positive for all $l\geq 2$, but $\iclos(s)=E(G)$ for some $s\neq E(G)$ and $\omega(E(G))=0$. Thus, a more meticulous treatment is required to make the powers positive.

\subsection{Influence of the magnetic moment projector and making the degrees positive}\label{subsec_projector}

We will estimate the influence of $\prj$ for the case when $|P|=0$ (see the definitions in Section \ref{subsec_numerator}). It is more convenient to consider the Feynman amplitude as a function of $p_1,p_2$, where $p_1=p-q/2$, $p_2=p+q/2$. We will prove the following two lemmas for this estimation.
\begin{lemma}\label{lemma_amm_mpp}
The following equality is satisfied:
$$
\prj\left[\prod_{j=l}^1\left(\gamma_{\xi_j}(m+\hat{p}_2)\right) \cdot \gamma_{\mu} \cdot \prod_{j=1}^n \left((m+\hat{p}_1)\gamma_{\nu_j}\right)\right] =0,
$$
where $\xi_1,\ldots,\xi_l,\nu_1,\ldots,\nu_n$ are \emph{names} of tensor indices; each name occurs twice in the selection; the convolution is performed over pairs of the same tensor indices using $g_{\mu\nu}$.
\end{lemma}
\begin{proof}
First, let us note that in the definition of $\prj$ we consider only expressions of the form $\overline{\psi}_2 \Gamma(p_1,p_2)\psi_1$, where 
$$
(m-\hat{p}_1)\psi_1=0,\quad \overline{\psi}_2 (m-\hat{p}_2)=0,\quad p_1^2=p_2^2=m^2.
$$
We will use the formula
\begin{equation}\label{eq_lemma_amm_mpp_mmp}
\gamma_{\xi_l}(m+\hat{p}_2)=(m-\hat{p}_2)\gamma_{\xi_l}+2p_{2\xi_l}.
\end{equation}
The part of the expression corresponding to the first term of (\ref{eq_lemma_amm_mpp_mmp}) right part is cancelled, as it was noted above. For considering the second term, we perform the convolution of $p_{2\xi_l}$ with the corresponding multiplier $\gamma_{\ldots}$. The convolution changes this $\gamma_{\ldots}$ to the corresponding $p_{2\ldots}$. If this $p_{2\ldots}$ is to the left side of $\gamma_{\mu}$, then we use 
$$
(\hat{p}_2+m)\hat{p}_2(\hat{p}_2+m) = 2m^2 (\hat{p_2}+m)
$$
and reduce the problem to the statement of the lemma with lesser $l,n$. If this $p_{2\ldots}$ occurred in the begin of the expression, we use that 
$$
\prj[\Gamma]=\frac{1}{2m}\prj[(\hat{p_2}+m)\Gamma]
$$ in addition (this follows from the note in the beginning). The case when the $p_{2\ldots}$ occurs to the right side of $\gamma_{\mu}$ is considered analogously, using
$$
(\hat{p}_1+m)\hat{p}_2(\hat{p}_1+m) = (2p_1p_2)(\hat{p_1}+m),\quad \prj[\Gamma]=\frac{1}{2m}\prj[\Gamma(\hat{p_1}+m)]
$$
and the fact that $2p_1p_2=2m^2$ can be factorized out of $\prj[\ldots]$ (see Section \ref{subsec_constr_feynman_amm}).

If $l=0$, then we perform the analogous transformation from the right end of the product inside $\prj[\ldots]$. If $l=0,n=0$, we use $\prj\gamma_{\mu}=0$. This completes the proof.
\end{proof}

\begin{lemma}\label{lemma_amm}
If $|P|=0$, then the following inequality is satisfied in terms of (\ref{eq_global_multiplier_factor_out}):
\begin{equation}\label{eq_lemma_amm}
\left| \sum_{j:P_j=P} [\prj \Pi_j] Y_j(z) \right| \leq C \cdot \max_{i\in \photon(E(G))} \frac{z'_i}{\max(z'_i,z_i)},
\end{equation}
where $C$ is some constant that depends only on the structure of the graph (and $m$), $z'_i$ are defined in Lemma \ref{lemma_w}.
\end{lemma}
\begin{proof}
Let us enumerate the lepton lines of $G$ along the path: $1,2,\ldots 2h$. Then the sum in (\ref{eq_lemma_amm}) can be expressed as
$$
X=\prj\left[ \gamma_{\ldots} (m+\hat{Q}'_{2h}) \gamma_{\ldots} \ldots (m+\hat{Q}'_{h+1}) \gamma_{\ldots} \gamma_{\mu} \gamma_{\ldots} (m+\hat{Q}'_{h}) \gamma_{\ldots} \ldots (m+\hat{Q}'_{1}) \gamma_{\ldots} \right],
$$
where $\hat{Q}'_l=\hat{Q}_l(z)/D(z)$ (see Section \ref{subsec_constr_schwinger}). Put
$$
Q''_l=Q'_l-q_l,
$$
where $q_l=p_1$ for $1\leq l\leq h$ and $q_l=p_2$ for $h+1\leq l\leq 2h$. Taking into account Lemma \ref{lemma_amm_mpp} and $m+\hat{Q}'_l=(m+\hat{q}_l)+\hat{Q}''_l$, we obtain that $X$ can be expressed as the sum of analogous expressions $\prj[\ldots]$ with multipliers $\hat{Q}''_l$ or $(m+\hat{q}_l)$ instead of $(m+\hat{Q}'_l)$ and with at least one multiplier $\hat{Q}''_l$.

Let us fix the selection of the sum term and the line $l$ with the multiplier $\hat{Q}''_l$. By $X'$ we denote the contribution of this term. All other multipliers (except $\gamma_{\ldots}$) are linear combinations of $\hat{p}_1,\hat{p}_2$ and $1$ with coefficients that are less or equal $1$ (in absolute value); thus, it is sufficient to estimate the coefficients of $Q''_l$. We have
$$
Q''_l=\frac{Q_l(z)-q_l D(z)}{D(z)}.
$$
Both terms of the numerator can be expressed as sums of the form $\sum_T c(T)\prod_{l\in E(G)\backslash T} z_l$,
where the summation goes over 1-trees $T$ of $G$, the coefficients $c(T)$ are linear combinations of $p_1,p_2$. The terms corresponding to $T$ are cancelled if $l\in T$ and the momentum passing through $l$ in $T$ equals $q_l$ (see Section \ref{subsec_constr_schwinger}).

Suppose $T$ is not cancelled. By definition, put
$$
X'_T=\frac{\prod_{l\in E(G)\backslash T} z_l}{D(z)}.
$$
It is obvious that there exists $i\in\photon(E(G))$ such that $i\in T$ (because $T=\lepton(E(G))$ is cancelled).  Let us consider the vertexes $v_1,\ldots,v_n$ of the lepton path connecting the ends of $i$ ordered along it. The ends $v_1$ and $v_n$ of $i$ belong to different connectivity components of $T\backslash i$. Therefore, there exists $b$ such that $v_b,v_{b+1}$ belong to the different connectivity components. By $j$ we denote the lepton line connecting $v_b$ and $v_{b+1}$. It is obvious that $j\notin T$ and $T'=T\cup\{j\}\backslash \{i\}$ is a 1-tree of $G$. From the fact that both $T$ and $T'$ contribute to $D(z)$, it follows that
$$
X'_T \leq \frac{1}{1+\frac{z_i}{z_j}} = 1-\frac{z_i}{z_i+z_j} \leq 1-\frac{z_i}{z_i+z_i'} = \frac{z_i'}{z_i+z_i'} \leq \frac{z_i'}{\max(z_i,z_i')}.
$$
This completes the proof.
\end{proof}

Let us denote by $i_0$ the value of $i$ on which the maximum (\ref{eq_lemma_amm}) is reached. As a consequence of Lemma \ref{lemma_amm}, we have the following inequality in terms of Lemma \ref{lemma_a_b}:
\begin{equation}\label{eq_a_b_c}
\left| \frac{K(z)\sum_{j:P_j=P} [A\Pi_j] Y_j(z)}{W(z)^M} \right| \leq C\cdot \frac{\prod_{l=1}^L t_l^{\lceil-\omega(\iclos(s^{[l]}))\rceil+A_l+B_l+C_l}}{z_1\ldots z_L},
\end{equation}
where
\begin{equation}\label{eq_cl}
C_l=\begin{cases}
1,\text{ if } |P|=0 \text{ and } i_0\in\iclos(s^{[l]})\backslash s^{[l]}, \\
0\text{ in the other cases.}
\end{cases}
\end{equation}

It is time to formulate and prove the main theorem.
\begin{theorem}
For the Feynman parametric integrand $I'(z_1,\ldots,z_L)$ we have
$$
\left| I'(z_1,\ldots,z_L) \right| \leq C\cdot \frac{\prod_{l=2}^L t_l^{\max\left( \lceil-\omega(\iclos(s^{[l]}))\rceil - \frac{1}{2}, \frac{1}{2} \right)}}{z_1\ldots z_L},
$$
where $C$ is some constant depending only on the structure of the graph. Let us remark that we ignore the multiplier containing $t_1$ (because we use the Feynman parameters, $1/L\leq t_1\leq 1$).
\end{theorem}
\begin{proof}
Let us fix $P$, $i_0$ (for $|P|=0$) and define $r_i$ for $i\in\photon(E(G))$ using Lemma \ref{lemma_r}. We will use (\ref{eq_a_b_c}), but with $A_l',B_l',C_l'$ instead of $A_l,B_l,C_l$, where $A_l',B_l',C_l'$ are defined by the rules (\ref{eq_al}), (\ref{eq_bl}), (\ref{eq_cl}), but with $r_i'$ instead of $r_i$. We will define $r_i'$ below as a modification of $r_i$. 

Let us choose $j_0\in\photon(E(G))$ such that $j_0\notin s^{[l]}$ for all $l$ satisfying $l\geq 2$, $\lepton(E(G))\subseteq s^{[l]}$ (it is sufficient to examine the minimal $l$ satisfying this). 

To prove the theorem it is enough to obtain real numbers $r'_i\geq 0$ for $i\in \photon(E(G))$ such that $\sum_i r'_i = |\photon(E(G))|-|P|$ and $A_l'+B_l'+C_l'\geq 1/2$ for all $l$ satisfying $\lepton(E(G))\subseteq s^{[l]}$, $A_l'+B_l'+C_l'\geq -1/2$ for the other $l$ ($l\geq 2$).

There are two cases:
\begin{enumerate}
\item $|P|\geq 1$. This case splits into two subcases:
\begin{itemize}
\item $r_{j_0}=0$; in this case, we put $r_i'=r_i$ for all $i$ and we have $A_l'\geq 0$, $C_l'\geq 0$; $B_l'\geq 1$ for $\lepton(E(G))\subseteq s^{[l]}$, $B'_l\geq 0$ for the other cases;
\item $r_{j_0}=1$; let us take $j$ such that $r_j=0$ (we always can do this, because $\sum_i (1-r_i)=|P|>0$); put $r'_{j_0}=1/2$, $r'_{j}=1/2$, $r'_i=r_i$ for the other $i$; we have $A_l'=0$, $B_l'\geq 1/2$, $C_l'\geq 0$ for $\lepton(E(G))\subseteq s^{[l]}$, $A_l'\geq A_l-1/2\geq -1/2$, $B_l'\geq 0$, $C_l'\geq 0$ for the other cases. 
\end{itemize}
\item $|P|=0$. This case splits into subcases too:
\begin{itemize}
\item $i_0=j_0$; in this case, we put $r_i'=r_i$ for all $i$ and we have $A_l'\geq 0$, $B_l'\geq 0$; $C_l=1$ for $\lepton(E(G))\subseteq s^{[l]}$, $C_l\geq 0$ for the other cases;
\item $i_0\neq j_0$, $r_{j_0}=0$; put $r_i'=r_i$ for all $i$ and we have $A_l'\geq 0$, $C_l'\geq 0$; $B_l'\geq 1$ for $\lepton(E(G))\subseteq s^{[l]}$, $B'_l\geq 0$ for the other cases;
\item $i_0\neq j_0$, $r_{j_0}=1$; put $r'_{j_0}=1/2$, $r'_{i_0}=r_{i_0}+1/2$, $r'_i=r_i$ for the other $i$; we have $A'_l=0$, $B'_l+C'_l\geq 1/2$ for $\lepton(E(G))\subseteq s^{[l]}$ (if $i_0\in\iclos(s^{[l]})\backslash s^{[l]}$, then $C'_l=1$ will compensate $r'_{i_0}>1$ if occur), $A'_l\geq A_l-1/2\geq -1/2$, $B'_l+C'_l\geq 0$ for the other cases (analogously).
\end{itemize}
\end{enumerate}
The theorem is proved.
\end{proof}

\section{Conclusions}

The finiteness was proved rigorously for each contribution to the lepton anomalous magnetic moment of QED Feynman graphs without lepton loops and without UV divergent subgraphs\footnote{If a graph contains UV divergent subgraphs, divergence subtraction is required; this complicates the problem. The way based on the subtraction procedure from ~\cite{volkov_2015} appears to be promising. This procedure eliminates all IR and UV divergences in Feynman parametric space before integration without any use of dimensional regularization and so on. Moreover, it is based on linear operators and has a simple and elegant formulation. This method was used for obtaining high-precision results ~\cite{volkov_5loops_prd,volkov_gpu}.}. This finiteness was proved in terms of Feynman parameters; the rigorous part starts from the point when the Feynman parametric integral is constructed. For proving the finiteness an upper bound for the absolute value of the Feynman parametric integrand was obtained. This upper bound is formulated in terms of the Hepp sectors, ultraviolet degrees of divergence, I-closures. 

A flexibility of the upper bound in different Hepp sectors allows us to use this idea for numerical integration based on the Monte Carlo method. The notion of I-closure was introduced by the author in the previous papers for numerical calculation of the electron $g-2$. The ideas and technique of this partial case may be, in principle, applied to other problems.

As a consequence, the inequality (\ref{eq_general_est}) for the Feynman parametric integrand was obtained. The power $1/2$ in the inequality is not improvable (see Appendix) in contrast to the ``pure ultraviolet'' case for which we can make do with integer numbers. Perhaps this means that approximations in E. Speer's form (\ref{eq_speer}) are not adequate for properly examining the infrared limit. However, no other estimations are known.

It is very important to note that there is no mathematically rigorous general case divergence cancellation proof even for QED at this point in time\footnote{There are a lot of calculations demonstrating this cancellation. However, all known attempts to make a general-case mathematical proof suffer from drawbacks and incompleteness (as mentioned in Section \ref{sec_intro}).}.

\section*{Acknowledgments}

The author thanks Lidia Kalinovskaya, Oleg Teryaev and Andrey Kataev for their help in organizational issues.

\section*{Appendix: an example that demonstrates the exactness of the estimation}

Let us consider an example demonstrating that the power $1/2$ in (\ref{eq_general_est}) is not improvable. Let us take the two-loop Feynman graph from Fig. \ref{fig_example_unimprov}.

\begin{figure}[h]
\begin{center}
\includegraphics{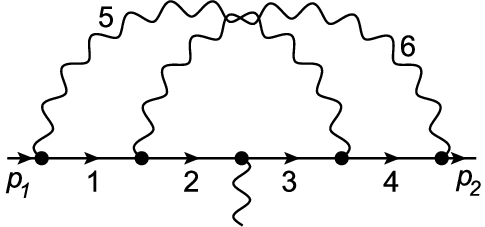}
\end{center}
\caption{Feynman graph for demonstrating the exactness of the estimation.}
\label{fig_example_unimprov}
\end{figure}

For this graph we will use the following values of Feynman parameters:
$$
z_1=\Lambda\delta,\ z_2=z_3=z_4=\delta^2,\ z_5\sim 1,\ z_6=\Lambda\delta^2,
$$
where $\delta\rightarrow 0$, for the constant $\Lambda$ we take some sufficiently big value. We will write $f(\delta)\sim g(\delta)$, if $\lim_{\delta\rightarrow 0} f(\delta)/g(\delta)=1$. Also, we will write $f(\delta)\asymp g(\delta)$, if $0<C_1\leq |f(\delta)/g(\delta)|<C_2$ for all $0<\delta<\delta_0$, where $\delta_0>0$.

We have
$$
D\sim (\Lambda+3)\delta^2,\ W\sim -m^2\left(\frac{9}{\Lambda+3}+\Lambda^2\right)\delta^2,
$$
$$
B_{12}=B_{13}=(\Lambda+1)\delta^2,\ B_{14}=-2\delta^2,\ B_{23}=B_{24}=B_{34}\sim 1,
$$
$$
\min(z_1,\ldots,z_6)\asymp \delta^2,\ z_1\ldots z_6\asymp \delta^9.
$$
We have terms that potentially can have the asymptotics
$$
\frac{\min(z_1,\ldots,z_6)^{1/2}}{z_1\ldots z_6}\asymp \frac{1}{\delta^8}
$$
for $P=\{\empty\}$, $P=\{\{2,3\}\}$, $P=\{\{2,4\}\}$, $P=\{\{3,4\}\}$ (we have $M=2-|P|$ in (\ref{eq_feynman_param_term})). The most difficult task is to demonstrate that the asymptotics is not cancelled in some way.

If we take a sufficiently big $\Lambda$, then the ``naive'' asymptotics for the terms corresponding to $|P|=1$ will dominate over the other terms asymptotics. Moreover, in the multipliers $m+\hat{Q}_l/D$ for $l=1,2,3,4$ the part $m+\hat{p}_j$ will dominate, where $j=1$ for $l=1,2$, $j=2$ for $l=3,4$. The contributions of these dominated terms are
$$
\frac{B_{23}}{D^3 W} \prj[F_{23}],\quad \frac{B_{24}}{D^3 W} \prj[F_{24}],\quad \frac{B_{34}}{D^3 W} \prj[F_{34}]
$$
for $P=\{\{2,3\}\}$, $P=\{\{2,4\}\}$, $P=\{\{3,4\}\}$ correspondingly, where
$$
F_{23}=\gamma_{\nu}(m+\hat{p}_2)\gamma_{\lambda}\gamma_{\xi}\gamma_{\mu}\gamma_{\xi}\gamma_{\nu}(m+\hat{p}_1)\gamma_{\lambda},
$$
$$
F_{24}=\gamma_{\nu}\gamma_{\xi}\gamma_{\lambda}(m+\hat{p}_2)\gamma_{\mu}\gamma_{\xi}\gamma_{\nu}(m+\hat{p}_1)\gamma_{\lambda}
$$
$$
F_{34}=\gamma_{\nu}\gamma_{\xi}\gamma_{\lambda}\gamma_{\xi}\gamma_{\mu}(m+\hat{p}_1)\gamma_{\nu}(m+\hat{p}_1)\gamma_{\lambda}.
$$
The calculation with the help of a computer gives
\small
$$
F_{23}=8m\gamma_{\mu}\hat{p}_1 + 8m\hat{p}_2\gamma_{\mu} - 16mp_{2\mu} - 16mp_{1\mu} + 16(p_1p_2)\gamma_{\mu} - 8m^2\gamma_{\mu},
$$
$$
F_{24}=16(p_1p_2)\gamma_{\mu}-16\hat{p}_2p_{1\mu}+32m p_{1\mu}+16m^2 \gamma_{\mu} - 16m \gamma_{\mu}\hat{p}_1 +16m p_{2\mu}+8\hat{p}_2 \gamma_{\mu} \hat{p}_1 - 16\hat{p}_1 p_{2\mu},
$$
$$
F_{34} = -8m \gamma_{\mu} \hat{p}_1 + 16mp_{1\mu} + 8p_1^2 \gamma_{\mu} - 16p_{1\mu} \hat{p}_1 - 8m^2\gamma_{\mu},
$$
$$
\prj[F_{23}]=16m^2,\quad \prj[F_{24}]=-8m^2,\quad \prj[F_{34}]=0.
$$
Since $\prj[F_{23}+F_{24}+F_{34}]=8m^2\neq 0$, $B_{23}=B_{24}=B_{34}$, these asymptotically dominated terms are not cancelled.

\end{document}